\documentclass[a4paper,UKenglish,cleveref,
autoref]{lipics-v2019}
\PassOptionsToPackage{labelformat=simple}{subcaption}
\usepackage{stmaryrd}
\usepackage{booktabs}
\usepackage{tikz}
\usetikzlibrary{shapes.multipart}
\usetikzlibrary{decorations.pathreplacing}

\title{Query Preserving Watermarking Schemes for Locally Treelike
Databases} 

\titlerunning{Query Preserving Watermarking Schemes for Locally Treelike
Databases}

\author{Agnishom Chattopadhyay}{Chennai Mathematical Institute, India
\and UMI ReLaX, Indo-French joint research unit}{}{}{}

\author{M.~Praveen}{Chennai Mathematical Institute, India
\and UMI ReLaX, Indo-French joint research unit}{}{}{Partially supported by a grant from the Infosys foundation.}

\authorrunning{Agnishom Chattopadhyay and M.~Praveen}

\Copyright{Agnishom Chattopadhyay and M.~Praveen}



\category{}

\supplement{}

\acknowledgements{The authors thank B.~Srivathsan and K.~Narayan Kumar
for feedback on the draft.}

\nolinenumbers 

\hideLIPIcs  

\EventEditors{John Q. Open and Joan R. Access}
\EventNoEds{2}
\EventLongTitle{42nd Conference on Very Important Topics (CVIT 2016)}
\EventShortTitle{CVIT 2016}
\EventAcronym{CVIT}
\EventYear{2016}
\EventDate{December 24--27, 2016}
\EventLocation{Little Whinging, United Kingdom}
\EventLogo{}
\SeriesVolume{42}
\ArticleNo{23}

\theoremstyle{definition}

\begin{document}

\maketitle

\begin{abstract}
Watermarking is a way of embedding information in digital documents.
Much research has been done on techniques for watermarking relational
databases and XML documents, where the process of embedding
information shouldn't distort query outputs too much. Recently, techniques
have been proposed to watermark some classes of relational structures
preserving first-order and monadic second order queries. For
relational structures whose Gaifman graphs have bounded degree,
watermarking can be done preserving first-order queries.

We extend this line of work and study watermarking schemes for other
classes of structures. We prove that for relational structures whose
Gaifman  graphs belong to a class of graphs that  have locally bounded
tree-width and is closed under minors, watermarking schemes exist that
preserve first-order queries. We use previously known properties of
logical formulas and graphs, and build on them with some technical
work to make them work in our context.  This constitutes a part of the
first steps to understand the extent to which techniques from
algorithm design and computational learning theory can be adapted for
watermarking.
\end{abstract}

\section{Introduction}
\label{sec:introduction}
Watermarking of digital content can be used to check intellectual
property violations. The idea is to embed some information, such as a
binary string, in the digital content in such a way that it is not
easily apparent to the end user. If the legitimate owner of the digital content
suspects a copy to be stolen, they should be able to retrieve the
embedded information, even with limited access to the stolen copy,
even if it has been tampered to remove the embedded information. Here
there are two opposing goals. One is to be able to embed large amount
of information. The other is to ensure that the embedding doesn't
distort the content too much.

There can be many ways to measure how much distortion is acceptable.
In \cite{AK2002}, embedding is performed by flipping bits in numerical
attributes while preserving the mean and variance of all numerical
attributes. There are other works that focus on the
specific use of the digital content: in \cite{KZ2000}, the digital
content consists of graphs whose vertices represent locations and
weighted edges represent distance between locations. It is shown that
information can be embedded in such a way that the shortest distance
between any two locations is not distorted too much.

We study embedding information in relational databases such that the
distortion on query outputs is bounded.
\begin{example}
\label{ex:employee}
  The table \texttt{EmployeeTable} shown in Table~\ref{tb:EmployeeTable} is an
example of a database instance of an organization's record of employees.

\begin{table}
\begin{subtable}[t]{0.45\textwidth}
  \caption{The original \texttt{EmployeeTable}}
  \label{tb:EmployeeTable}
  \begin{tabular}{llr}
  \toprule
  FirstName &  City & Salary \\
  \midrule
  John &  Chennai & 10,000 \\
  Arjun &  Coimbatore & 20,000 \\
  Pooja &  Chennai & 15,000 \\
  Neha &  Vellore & 30,000 \\
  Padma &  Coimbatore & 20,000 \\
  \bottomrule
  \end{tabular}
\end{subtable}
\begin{subtable}[t]{0.45\textwidth}
  \caption{A distorted \texttt{EmployeeTable}}
  \label{tb:EmployeeTableD}
  \begin{tabular}{llr}
  \toprule
  FirstName &  City & Salary \\
  \midrule
  John &  Chennai & 10,001 \\
  Arjun &  Coimbatore & 19,999 \\
  Pooja &  Chennai & 14,999 \\
  Neha &  Vellore & 30,000 \\
  Padma &  Coimbatore & 20,001 \\
  \bottomrule
  \end{tabular}
\end{subtable}
  \caption{\texttt{EmployeeTable} of Ex.~\ref{ex:employee}}
\end{table}

  Consider the following query parameterized by the variable $x$.
\begin{align*}
  \psi(x) \equiv \texttt{select FirstName,Salary from EmployeeTable where City=}x
\end{align*}
  If we substitute the variable $x$ with a particular city $c$, the
above query lists the salaries of individuals working in that city.
Let $\mathtt{total}(c)$ be the sum of all salaries listed by the query
$\psi(c)$. We want to hide data in \texttt{EmployeeTable} by
distorting the \texttt{Salary} field. Let $\mathtt{total}'(c)$ be the
sum of all salaries listed by the query $\psi(c)$ run on the distorted
database. We say that the distortion \emph{preserves the query
$\psi(x)$} if there is a constant $B$ such that for any city $c$, the
absolute value of the difference between $\mathtt{total}(c)$ and
$\mathtt{total'}(c)$ is bounded by $B$.

Assuming that we can distort each employee's salary by at most $1$
unit and we wish to maintain the bound $B$ to be $0$, our options
are the following: increase the salary of \texttt{John} by $1$ and
decrease the salary of \texttt{Pooja} by $1$ or vice-versa.
Similarly for \texttt{Arjun} and \texttt{Padma}. This gives us four
different ways distort the data base. 
These distortions are designed to preserve only the query $\psi(x)$.
If a different query is run on the same distorted databases, the
results may vary widely.
Suppose the organization
distributes  the four distorted databases among it's branches. If a
stolen copy of the database is found, the
organization can run the query $\psi(x)$ on the stolen copy. By
observing the salaries of \texttt{John}, \texttt{Pooja},
\texttt{Arjun} and \texttt{Padma} and comparing them with the values
from the original database, one can narrow down on the branches where
the leakage happened. The organization only needs to run the query
$\psi(x)$ on the suspected stolen copy, just like any normal consumer
of the database. This is important since the entity possessing the
stolen copy may not allow full access to its copy of the database. We
say a watermarking scheme is \emph{scalable} if for larger databases,
there are larger number of ways to distort the database, while still
preserving queries of interest.
\end{example}

In \cite{Gross2011}, meta theorems are proved regarding the existence
of watermarking schemes for classes of databases preserving queries
written in classes of query languages. The Gaifman graph of a database
is a graph whose set of vertices is the set of elements in the
universe of the database. There is an edge between two elements if the
two elements participate in some tuple in the database.  For
databases with unrestricted structures, even simple queries can't be
preserved; see \cite[Theorem 3.6]{Gross2011} and \cite[Example
3]{GT2004}.  Preserving queries written in powerful query languages
and handling databases with minimum restrictions on their structure
are conflicting goals. For databases whose Gaifman graphs have bounded
degree, first-order queries can be preserved \cite{Gross2011}. It is
also shown that for databases whose Gaifman graphs are similar to
trees, MSO queries can be preserved. The similarity of a graph to a
tree is measured by tree-width. For example, XML documents are trees
and have tree-width $1$.

\subparagraph*{Contributions} We prove that watermarking schemes exist
for databases whose Gaifman graphs belong to a class of graphs that
have locally bounded tree-width and is closed under minors,
preserving unary first-order queries. Classes of graphs with
bounded degree are contained in this class. A graph $G$ has locally
bounded tree-width if it satisfies the following property: there
exists a function $f$ such that for any vertex $v$ and any number $r$,
the sphere of radius $r$ around $v$ induces a subgraph on $G$ whose
tree-width is at most $f(r)$.

\subparagraph*{Why first-order logic?} The pivotal Codd's theorem \cite{Codd1972}
states that first-order logic is expressively equivalent to relational
algebra, and relational algebra is the basis of standard relational
database query languages.
\subparagraph*{Why locally bounded
treewidth?} Classes of graphs with locally bounded treewidth are
good starting points to start using techniques from algorithm
design and computational learning theory in other areas.
Seese \cite{Seese1996} proved that first-order properties can be
decided in linear time for graphs of bounded degree. Baker
\cite{Baker1994} showed efficient approximation algorithms for some
specific hard problems, when restricted to planar graphs.  Eppstein
\cite{Eppstein2000} showed that Baker's technique continues to work in
a bigger class of graphs: it suffices for the class of graphs to have
locally bounded tree-width and additionally, the class should be closed
under minors. Frick and Grohe \cite{FG2001} showed that on any class
of graphs with locally bounded tree-width, any problem definable in
first-order logic can be decided efficiently\footnote{Here, efficiency
means fixed parameter tractability; see \cite{FG2001} for details.}.
For problems definable in first-order logic, the classes of graphs for
which efficient algorithms exist was then extended to bigger and
bigger classes: excluded minors \cite{FG2001J}, locally excluded
minors \cite{DGK2007}, bounded expansion, locally bounded expansion
\cite{DKT2010} and nowhere dense \cite{GKS2017}. It is now known that
nowhere dense graphs are the biggest class of graphs for which there
are efficient algorithms for first-order definable problems
\cite{DKT2010,Kreutzer2011,GKS2017}, provided some complexity
theoretic assumption are true. Results related to computational
learning theory have been proved in \cite{GT2004} for classes of
graphs with locally bounded treewidth. Recently, similar results have
been proven for nowhere dense classes of graphs \cite{PST2018}.
\subparagraph*{Why unary queries?} Some of the techniques we have used
are difficult to extend to non-unary queries. Some technical details
about this are discussed in the conclusion.

\subparagraph*{Related Works} The fundamental definitions of what it
means for a watermarking scheme to be scalable and preserving a
query was given in \cite{KZ2000}.
It was shown in \cite{KZ2000} that on weighted graphs, scalable
watermarking schemes exist preserving shortest distance between
vertices. The adversarial model was also introduced in \cite{KZ2000},
where the person possessing the stolen copy can introduce additional
distortion to evade detection. It is shown in \cite{KZ2000} that a
watermarking scheme for the non-adversarial model can be transformed
to work for the adversarial model, under some assumption about the
quantity of distortion introduced by the person trying to evade
detection, and the amount of knowledge the person possesses.
Gross-Amblard \cite{Gross2011} adapted these definitions for
relational structures of any vocabulary and any query written in
Monadic Second Order (MSO) logic, and showed results about classes of
structures of bounded degree and tree-width. Gross-Amblard \cite{Gross2011} also
provided the insight that existence of scalable watermarking schemes
preserving queries from a certain language is closely related to
learnability of queries in the same language. We make use of this
insight in our work. Grohe and Tur{\'a}n
\cite{GT2004} proved that MSO-definable families of sets in graphs of
bounded tree-width have bounded Vapnik-Chervonenkis (VC) dimension,
which has well known connections in computational learnability theory.
It is also shown in \cite{GT2004} that on classes of graphs with
locally bounded tree-width, first-order definable families of sets
have bounded VC dimension.

\section{Preliminaries}
\label{sec:preliminaries}
\subparagraph*{Relational databases} A signature (or
database schema) $\tau$ is a finite set of relation symbols
$\{R_1, \ldots ,R_t \}$. We denote by $r_{i}$ the arity of
$R_{i}$ for every $i \in \{1, \ldots, t\}$. A $\tau$-structure $G =
( V, R_1^{G}, \ldots R_t^{G} )$ (or database
instance) consists of a set $V$ called the universe, and an
interpretation $R_{i}^{G} \subseteq V^{r_{i}}$ for every relation symbol
$R_{i}$. For a fixed $s \in
\mathbb{N}$, a \emph{weighted structure} $(G,
W)$ is a finite structure $G$ together with a
weight function $W$, which is a partial function from
$V^s$ to $\mathbb{N}$, that maps a $s$-tuple $\overline{b}$
to its weight $W(\overline{b})$.

\subparagraph*{First Order and Monadic Second Order Queries}
An atomic formula is a formula of the form $x = y$ or $R(x_1,
\ldots x_r)$, where $x, y, x_1 \ldots x_r$ are variables
and $R$ is an $r$-ary relational symbol in $\tau$. First-order (FO)
formulas are formulas built from atomic formulas using the usual
boolean connectives and existential and universal quantification over the elements
of the universe of a structure.

Monadic Second Order (MSO) logic extends first-order logic by allowing
existential and universal quantifications over subsets of the
universe. Formally, there are two types of variables. Individual
variables, which are interpreted by elements of the universe of a
structure, and set variables, which are interpreted by subsets of the
universe of a structure. In addition to the atomic formulas of
first-order logic mentioned in the previous paragraph, MSO has atomic
formulas of the form $X(x)$, saying that the element interpreting the
individual variable $x$ is in the set interpreting the set variable
$X$. Furthermore, MSO has quantification over both individual and set
variables.

The quantifier rank, denoted $\mathrm{qr}(\psi)$ of a formula $\psi$
is the maximum number of nested quantifiers in $\psi$.  A free
variable of a formula $\psi$ is a variable $x$ that does not occur in
the scope of a quantifier.  The set of free variables of a formula
$\psi$ is denoted by $\mathrm{free}(\psi)$. A sentence is a formula
without free variables. We write $\psi(x_{1}, \ldots, x_{r})$ to
indicate that $\mathrm{free}(\psi) \subseteq \{x_{1}, \ldots,
x_{r}\}$. We denote the size of $\psi$ by $||\psi||$. We only work
with formulas that have free individual variables, but not free set
variables. Given a vector $\overline{x}=\langle x_{1}, \ldots, x_{s}
\rangle$ of variables, a formula $\psi(\overline{x})$ and a structure
$G$, we denote by $\psi(G)=\{\overline{a} \in V^{s} \mid G \models
\psi(\overline{a})\}$ the set of tuples of elements from the universe
$V$ of $G$ that can be assigned to the variables $\overline{x}$ to
satisfy $\psi(\overline{x})$.

Suppose $\psi(\overline{x}, \overline{y})$ is a
formula with two distinguished vectors of free variables $\overline{x}$ of
length $r$ and $\overline{y}$ of length $s$. We call
$\psi(\overline{x}, \overline{y})$ a $s$-ary query with
$r$ parameters. Given a structure $G$,
we call $\psi(\overline{a}, G) = \{ \overline{b} \in V^s \mid G \models
\psi(\overline{a}, \overline{b}) \}$ the \emph{output of the query
$\psi(\overline{x}, \overline{y})$ with parameter $\overline{a}$}. We refer
to $r$ (resp.~$s$), the length of $\overline{x}$
(resp.~$y$), as the \emph{input length} (resp.~the
\emph{output length}) of $\psi(\overline{x}, \overline{y})$. Given a
weighted structure $(G, W)$,
a parametric query $\psi(\overline{x}, \overline{y})$ and a parameter
$\overline{a}$, we extend the weight
function $W$ to weights of query outputs by defining $W(\psi(\overline{a},
G)) = \Sigma_{\overline{b} \in  \psi(\overline{a}, G)}
{W(\overline{b})}$. For a given structure $G$ and a query
$\psi(\overline{x}, \overline{y})$, we define $U = \bigcup_{\overline{a} \in
V^r}\psi(\overline{a}, G)$ to be the set of \emph{active tuples}.

\subparagraph*{Watermarking schemes} Suppose $c,d \in \mathbb{N}$. A weighted
structure $(G, W')$ is a \emph{$c$-local
distortion} of another weighted structure $(G, W)$ if for all $\overline{b}
\in V^s$, $|W'(b) - W(b)| \leq c$. The weighted
structure $(G, W')$ is a \emph{$d$-global distortion} of
$(G, W)$ with respect to a query $\psi(\overline{x},
\overline{y})$ if and only if, for all $\overline{a} \in V^r$,
$|W'(\psi(\overline{a}, G)) - W(\psi(\overline{a}, G))| \leq d$.

\begin{definition}[\cite{KZ2000,Gross2011}]
    \label{def:watermarkingSchemes}
Given a class of weighted structures $\mathcal{K}$ and a
query $\psi(\overline{x}, \overline{y})$, a \emph{
watermarking scheme preserving $\psi(\overline{x}, \overline{y})$} is
a pair of algorithms $\mathcal{M}$ (called the marker) and
$\mathcal{D}$ (called the detector) along with a function $f :
\mathbb{N} \to \mathbb{N}$ and a constant $d \in \mathbb{N}$ such
that:
\begin{itemize}
\item The marker $\mathcal{M}$ takes as input a weighted
    structure $(G, W) \in \mathcal{K}$ and a mark $\mu$, which is a
bit string of length $f(|U|)$, where $U$ is the set of active tuples. It outputs a weighted structure $(G, W_\mu) \in
    \mathcal{K}$ such that $(G, W_\mu)$ is a 1-local and $d$-global
    distortion of $(G, W)$ for the query $\psi(\overline{x},
\overline{y})$.
\item The detector $\mathcal{D}$ is given $(G, W)$, the original
    structure as input and has access to an oracle that runs queries
    of the form $\psi(\overline{x}, \overline{y})$ on $(G, W_\mu)$. The
    output of $\mathcal{D}$ is the hidden mark $\mu$.
\end{itemize}
\end{definition}

Intuitively, the marker takes a bit string and hides it in the
database by distorting weights. The detector detects the hidden mark
by observing the weights and comparing it with the original weights.
The term $f(|U|)$ denotes the length of the bit string that is hidden
in the database by the marker. We call a watermarking scheme
\emph{scalable} if the function $f$ grows at least as fast as some
fractional power asymptotically. For example, the scheme is scalable
if $f(n) = \sqrt{n}$ for all $n$, but not scalable if $f(n) = \log n$
for all $n$. We will mention later why scalability is important in
situations where adversaries try to erase watermarks.
Note that the algorithm $\mathcal{D}$ interacts with
the marked database $(G, W_\mu)$ only through $\psi(\overline{x},
\overline{y})$ queries. Hence, it is not
worthwhile distorting the weights of tuples that are not active.

Continuing Example~\ref{ex:employee}, the query $\psi(x)$ given there
can be written in First-order as $(\psi(\langle\mathit{city} \rangle,
\langle\mathit{name},\mathit{salary} \rangle) =
\mathtt{EmployeeTable}(\mathit{name},\mathit{city},\mathit{salary})$.
The set of active tuples is $U=\{\langle \text{John,
10000}\rangle,\langle \text{Arjun, 20000}\rangle,\langle
\text{Pooja, 15000} \rangle,\langle \text{Neha, 30000} \rangle,\langle
\text{Padma,20000}\rangle\}$. We can increase the salary of \texttt{John} by $1$ and
decrease the salary of \texttt{Pooja} by $1$ or vice-versa.
Similarly for \texttt{Arjun} and \texttt{Padma}. This gives $4$
different distortions that are $1$-local and $0$-global.
The marker algorithm can take a mark, which is a bit string of length
$2$, so there are $4$ possible marks. The marker can assign these $4$
marks to the $4$ possible distortions. The detector can observe the changes to the
salaries by querying the distorted copy and comparing the results with
the original database. The detector can compute the hidden mark by
accessing the assignment of the $4$ marks to the $4$ possible
distortions given by the marker. For any instance database of this
signature, we can pair off an employee of a city with another employee
in the same city and use one such pair to encode one bit of a
watermark to be hidden. If there are $n$ active tuples, we can encode
$\frac{n}{2}$ bits, assuming that there are at least two employees in
each city. For this watermarking scheme, the function
$f$ is defined as $f(n) = \frac{n}{2}$ and this is a scalable scheme.

Watermarking schemes can also be put in a context where there are
adversaries who know that there is some hidden mark and try to prevent
the detector algorithm from working properly, by distorting the
database further. Instead of the oracle running queries on $(G,
W_\mu)$, the queries are run on $(G,W'_\mu)$, which is a distortion of
$(G,W_\mu)$. The detector has to still detect the hidden mark $\mu$
correctly. Under some assumptions about the quantity of distortion
between $(G, W_\mu)$ and $(G, W'_\mu)$, watermarking schemes that work
in non-adversarial models can be transformed to work in adversarial
models; we refer the interested readers to \cite{KZ2000,Gross2011}.
The correctness of such transformations depend on probabilistic
arguments, where scalability helps. With bigger watermarks that are
hidden to begin with, there is more room to play around with the
distortions introduced by the adversaries.

\section{Watermarking schemes}
\label{sec:schemes}
In this section, we prove that scalable watermarking schemes exist for
some type of structures. First we prove that if the Gaifman graphs belong to 
a class of graphs with bounded tree-width, then scalable water marking
schemes exist preserving unary MSO queries. Then we prove that if the
Gaifman graphs belong to a class of graphs that is closed under minors
and that has locally bounded tree-width, then scalable water marking
schemes exist preserving unary FO queries.

\subsection{MSO Queries on Structures with Bounded Tree-width}
\subsubsection{Trees, Tree Automata and Clique-width}
We begin by reviewing some concepts and known results that are needed.

A binary tree is a $\{S_1, S_2, \preceq\}$-structure, where $S_1$,
$S_2$ and $\preceq$ are binary relation symbols. A tree $\mathcal{T} =
( T, S_1^\mathcal{T}, S_2^\mathcal{T}, \preceq^\mathcal{T}
)$ has a set of nodes $T$, a left child relation
$S_1^\mathcal{T}$ and a right child relation $S_2^\mathcal{T}$.
Relation $\preceq^\mathcal{T}$ stands for the transitive closure of
$S_1^\mathcal{T} \cup S_2^\mathcal{T}$, the tree-order relation. Given
a finite alphabet $\Sigma$, let $\tau(\Sigma) = \{S_1, S_2, \preceq \}
\cup \{ P_a | a \in \Sigma \}$ where for all $a \in \Sigma$, $P_a$ is
a unary symbol. A $\Sigma$-tree is a structure $\mathcal{T} = (
T,S_1^\mathcal{T}, S_2^\mathcal{T}, \preceq^\mathcal{T},
(P^\mathcal{T}_a)_{a\in \Sigma})$, where the restriction
$( T, S_1^\mathcal{T}, S_2^\mathcal{T}, \preceq^\mathcal{T}
)$ is a binary tree and for each $v \in T$ there exists exactly
one $a \in \Sigma$ such that $v \in P^\mathcal{T}_a$. We denote this
unique $a$ by $\sigma^\mathcal{T}(v)$. Intuitively, this represents
the labelling of nodes by letters from $\Sigma$ where
$\sigma^\mathcal{T}(v)$ is the label for the node $v$. We consider
trees with a finite number of pebbles placed on nodes. The pebbles are
considered to be distinct: pebble $1$ on node $v_{1}$ and pebble
$2$ on node $v_{2}$ is not the same as pebble $1$ on node
$v_{2}$ and pebble $2$ on node $v_{1}$. For some $k
\geq 1$, let $\Sigma_k = \Sigma \times \{0, 1\}^k$. This extended
alphabet denotes the position of the pebbles in the
tree. Suppose $\mathcal{T}$ is a $\Sigma$-tree and $k$ pebbles are
placed on the nodes $\overline{v} =
\langle v_1, \ldots, v_k \rangle$. Then
$\mathcal{T}_{\overline{v}}$ is the $\Sigma_k$-tree with the same
underlying tree as $\mathcal{T}$ and
$\sigma^{\mathcal{T}_{\overline{v}}}(u) = (\sigma^\mathcal{T}(u),
\alpha_1, \ldots, \alpha_k)$ where $\alpha_i = 1$ if and only if
$u=v_i$.

A $\Sigma$-tree automaton is a tuple $A = (Q, \delta, F)$
where $Q$ is a set of states and $F \subseteq Q$ are the accepting
states. The function $\delta : ((Q \cup \{*\})^2 \times \Sigma) \to Q$
is the transition function, where $*$ is a special symbol not in  $Q$. A
run $\rho : T \to Q$ of $A$ on a $\Sigma$-tree $\mathcal{T}$
is a function satisfying the following conditions.
\begin{itemize}
    \item If $v$ is a leaf then $\rho(v) = \delta(*, *,
        \sigma^\mathcal{T}(v))$.
    \item If $v$ has two children $u_1$ and $u_2$, then $\rho(v) =
        \delta(\rho(u_1), \rho(u_2), \sigma^\mathcal{T}(v))$.
    \item If $v$ has only a left child $u$ then $\rho(v) =
        \delta(\rho(u), *, \sigma^\mathcal{T}(v))$.
    \item Similarly if $v$ has only a right child.
\end{itemize}
If $v$ is the root of $\mathcal{T}$, a run $\rho$ of $A$ on
$\mathcal{T}$ is an accepting run if $\rho(v) \in F$.  A
$\Sigma_{(r+s)}$ tree automaton defines a $s$-ary query with $r$
parameters. We denote by $A(\overline{a}, \mathcal{T}) =
\{\overline{b} \in T^s \mid A \text{ has an accepting run on
} \mathcal{T}_{\overline{a}\overline{b}}\}$ the output of the query
$A$ on $\mathcal{T}$ with parameter $\overline{a}$.  It is
well known that MSO queries and tree automata have the same expressive
power.

\subparagraph*{Clique-width} A well-known notion of measuring the
similarity of a graph to a tree is its \emph{tree-width}. Many nice
properties of trees carry over to classes of structures of bounded
tree-width. For our purposes, we use \emph{clique-width}, a related
notion. It is well known that a structure of tree-width at most $k$
has clique-width at most $2^{k}$ \cite{CO2000}.

A $k$-colored structure is a pair $(G, \gamma)$ consisting
of a structure $G$ and a mapping $\gamma : V \to
\{1, \ldots , k\}$. A basic $k$-colored structure is a
$k$-colored structure $(G, \gamma)$ where
$|V| = 1$ and $R^G = \emptyset$ for all $R$. We
let $\Gamma_k$ be the smallest class of structures that contain all
basic $k$-colored structures and is closed under the following
operations:
\begin{itemize}
\item \emph{Union}: take two $k$-colored structures on disjoint
    universes and form their union.
\item \emph{$(i \to j)$ recoloring, for $1 \leq i, j \leq k$}: take a $k$-colored structure and recolor all vertices colored $i$ to $j$.
\item \emph{$(R, i_1, \ldots i_r)$-connecting, for every $r \geq 1$,
    every $r$-ary relation symbol $R$ and every $1 \leq i_1, \ldots i_r \leq k$}: take a
    $k$-colored structure $(G, \gamma)$ and add all tuples
    $\langle v_1, \ldots v_r \rangle \in V^r$ with $\gamma(v_j) = i_j$ for $1 \leq j
    \leq r$ to $\mathcal{R}^G$.
\end{itemize}

The \emph{clique-width} of a structure $G$, denoted by $\mathrm{cw}(G)$,
is the minimum $k$ such that there exists a $k$-coloring $\gamma :
V \to \{1, \ldots k \}$ such that $(G, \gamma) \in
\Gamma_k$.

For every $k$-colored structure $(G, \gamma ) \in \Gamma_{k}$ we can
define a binary, labeled parse-tree in a straightforward way. The
leaves of this tree are the elements of $G$ labeled by their color, and
each inner node is labeled by the operation it corresponds to. A
parse-tree (also called a clique decomposition) of a structure $G$ of
clique-width $k$ is a parse tree of some $(G, \gamma ) \in
\Gamma_{k}$. For the next lemma, it is important to note that if
$\mathcal{T}$ is a parse-tree for a structure $G$, then $V \subseteq
T$.
\begin{lemma}[{\cite[Lemma~16]{GT2004}}]
    \label{lem:cliqueDecompositionTransduction}
    Let $k \ge 1$. For every MSO formula $\varphi(\overline{x})$ there
    is a MSO formula $\tilde{\varphi}(\overline{x})$
    such that for every structure $G$ of clique-width $k$ and for every
    parse-tree $\mathcal{T}$ of $G$ we have $\varphi(G) =
    \tilde{\varphi}(\mathcal{T})$. Furthermore, there are
    constants $c, d$ (only depending on $k$ and the signature
    $\tau$) such that
    $|| \tilde{\varphi} || \le c || \varphi ||$ and
    $\mathrm{qr}(\tilde{\varphi}) \le \mathrm{qr}(\varphi) + d$.
\end{lemma}
\subsubsection{Watermarking Schemes to Preserve MSO Queries on
Structures With Bounded Tree-width}
Now we prove that there are scalable watermarking schemes that work for
structures from classes with bounded tree-width and preserve a given
MSO query. At a high level, the idea is the following: the given MSO
query is converted to an equivalent tree automaton. If the number of
active elements is large compared to the number of states in the
automaton, we can select pairs of elements that can't be distinguished
by the automaton. Either both the elements are in the output of the
query or none of them are. Hence, distorting one of them by a positive
amount and the other one by a negative amount will not contribute to
the global distortion.

We begin with the following lemma, which says that if a tree automaton
runs on a large tree, we can find large number of pairs of nodes that
are ``similar'' with respect to the automaton. A similar result
is proved and used in \cite{GT2004} to show that MSO-definable
families of sets in graphs of bounded tree-width have bounded
Vapnik-Chervonenkis (VC) dimension. The similarity of the following
result with that of \cite{GT2004} hints at some possible connections
between watermarking schemes and VC dimension.
\begin{lemma}
    \label{lem:treeaut}
    Let $A$ be a $\Sigma_{r+1}$ tree automaton with $m$ states. Let
    $\mathcal{T}$ be a $\Sigma$ tree. Suppose $Y \subseteq T$ is a set of
    nodes of $T$ with at least $2m^{m}+2$ elements\footnote{A similar result is
stated in \cite{Gross2011} with $4m$ elements, but there is an error in
the proof; see Appendix~\ref{app:errorInGross2011} for details.}. Then, there exists $n =
    \lfloor \frac{|Y|}{4m^{m}+4} \rfloor$ pairwise disjoint sets of nodes $V_1, V_2,
    \ldots, V_n \subseteq T$ and pairs $(b_i, b_i') \in (V_i \cap Y)^2$ of distinct nodes for
    all $i \in \{1, \ldots, n\}$ satisfying the following property: for
    every $\overline{a} = \langle a_1, a_2, \ldots , a_r\rangle \in T^r$
    and every $i \in \{1, \ldots, n\}$, if $\{a_{1}, a_{2}, \ldots,
    a_{r}\} \cap V_{i} = \emptyset$ then the runs of $A$ on
    $\mathcal{T}_{\overline{a}b_{i}}$ and
    $\mathcal{T}_{\overline{a}b_{i}'}$ label the tree roots with the same
    state.
\end{lemma}
\begin{proof}
    This proof is a variant of the proof of \cite[Lemma 7]{GT2004}.
    From the bottom up, take a minimal subtree of $\mathcal{T}$
    (minimal with respect to inclusion) that contains at least
    $m^{m}+1$ elements of $Y$. Let $U_{1}$ be the elements of
    $Y$ in this subtree; then the root of this subtree is
    $\mathrm{lca}(U_{1})$, the least common ancestor of nodes in $U_1$. As the
    tree is binary, it holds that $|U_{1}|<2m^{m}+2$. Remove this subtree, and
    repeat the same procedure as long as at least $m^{m}+1$ elements from $Y$ are
    still present in the remaining tree. We get at least $\alpha=\lfloor
    \frac{|Y|}{2m^{m}+2} \rfloor$ sets $U_{1}, U_{2}, \ldots, U_{\alpha}$ that are
    pairwise disjoint. The nodes $\mathrm{lca}(U_{i})$ are all distinct.

    Define a binary relation $F$ on $H=\{U_{1}, U_{2}, \ldots,
    U_{\alpha}\}$ to be the set of all pairs $(U_{i},U_{j})$ such that
    $\mathrm{lca}(U_{i})$ is an ancestor of $\mathrm{lca}(U_{j})$
    and there is no $k$ such that $\mathrm{lca}(U_{k})$ is a
    descendent of $\mathrm{lca}(U_{i})$ and an ancestor of
    $\mathrm{lca}(U_{j})$.  Then $(H,F)$ is a forest with $\alpha$
    vertices and thus with at most $\alpha-1$ edges. Therefore, at
    most $\frac{\alpha}{2}$ vertices of this forest have more than one
    child. Without loss of generality, we can assume that
    $U_{1},\ldots, U_{\beta}$ have at most one child, for some
    $\beta \ge \frac{\alpha}{2}$.

    If $U_{i}$ has no children, we let $V_{i}$ be the set of nodes
    of the subtree rooted at $\mathrm{lca}(U_{i})$. If $U_{i}$ has one
    child $U_{j}$, we let $V_{i}$ be the set of nodes of the subtree
    rooted at $\mathrm{lca}(U_{i})$ that are not nodes in the
    subtree rooted at $\mathrm{lca}(U_{j})$.

    Now we select pairs $(b_{i},b_{i}')$ as stated in the lemma.
    Suppose $U_{i}$ has no children. Since $U_{i}$ has $m^{m}+1$
    elements and the automaton $A$ has only $m$ states, we can choose
    $b_{i},b_{i}'$ from $U_{i}$ such that the runs of $A$ on
    $\mathcal{T}_{\overline{a}b_{i}}$ and
    $\mathcal{T}_{\overline{a}b_{i}'}$ label $\mathrm{lca}(U_i)$ with
    the same state, where $\overline{a} = \langle a_1, a_2, \ldots ,
    a_r\rangle \in T^r$ is any tuple such that $\{a_{1}, a_{2},
    \ldots, a_{r}\} \cap V_{i} = \emptyset$. Indeed, the nodes in
    $\overline{a}$ don't appear in the subtree rooted at
    $\mathrm{lca}(U_i)$. Hence the runs of $A$ on
    $\mathcal{T}_{\overline{a}b_i}$ and
    $\mathcal{T}_{\overline{a}b'_i}$ label $\mathrm{lca}(U_i)$ with
    states that only depend on $b_{i}$ and $b_{i}'$ respectively.
    Since $U_{i}$ has $m^{m}+1$ elements and the automaton $A$
    has only $m$ states, we can choose $b_{i},b_{i}'$ from $U_{i}$
    such that the runs of $A$ on $\mathcal{T}_{\overline{a}b_{i}}$ and
    $\mathcal{T}_{\overline{a}b_{i}'}$ label $\mathrm{lca}(U_i)$ with
    the same state.
    
    Suppose $U_{i}$ has one child $U_{j}$. Let $q$ be a state of
    $A$, $b$ be a node in $V_{i}$ and let $\mathcal{T}_{qb}$
    denote the subtree obtained by performing the following changes on the subtree rooted at
    $\mathrm{lca}(U_{i})$:
    \begin{enumerate}
        \item Remove the subtree rooted at $\mathrm{lca}(U_{j})$ and
            put in its place a tree whose root is labeled with the
            state $q$ by the run of $A$.
        \item Place a pebble on the node $b$ to get a
            $\Sigma_{1}$ tree.
    \end{enumerate}
    Suppose $Q$ is the set of states of the automaton $A$. Define the
function $f_b: Q \to Q$ as $f_b(q)=q'$, where $q'$ is the state
labeling the root of $\mathcal{T}_{qb}$ when $A$ runs on it. We can
think of placing a pebble on $b$ as inducing the function $f_b$. If
$\mathrm{lca}(U_j)$ is labeled by $q$, then $\mathrm{lca}(U_i)$ is
labeled by $f_b(q)$. Since there are at least $m^m+1$ nodes in $U_i$
and there are at most $m^m$ functions from $Q$ to $Q$, we can choose
nodes $b_i,b'_i$ from $U_i$ such that $f_{b_i}=f_{b'_i}$.

Finally, let us prove that the pairs $(b_i,b'_i)$ we have chosen
satisfy the condition stated in the lemma. If $U_i$ has no children,
then nodes in $\overline{a}$ don't appear in the subtree rooted at
$\mathrm{lca}(U_i)$. Hence the runs of $A$ on
$\mathcal{T}_{\overline{a}b_i}$ and $\mathcal{T}_{\overline{a}b'_i}$
label $\mathrm{lca}(U_i)$ with the same state, by choice of
$(b_i,b'_i)$. Hence, the runs of $A$ on
$\mathcal{T}_{\overline{a}b_i}$ and $\mathcal{T}_{\overline{a}b'_i}$
label the tree roots also with the same state.

Suppose $U_i$ has one child $U_j$. The nodes in $\overline{a}$ don't
appear in $V_i$, but they may appear in the subtree rooted at
$\mathrm{lca}(U_j)$. Suppose the runs of $A$ on
$\mathcal{T}_{\overline{a}b_i}$ and $\mathcal{T}_{\overline{a}b'_i}$
label $\mathrm{lca}(U_j)$ with some state $q$. Then
$\mathrm{lca}(U_i)$ is labeled by $f_{b_i}(q)$ and $f_{b'_i}(q)$
respectively. By choice, $f_{b_i}(q) = f_{b'_i}(q)$. Hence, the roots
of $\mathcal{T}_{\overline{a}b_i}$ and $\mathcal{T}_{\overline{a}b'_i}$
are labeled by the same state.
\end{proof}

The following result is proved in \cite{Gross2011}, but the proof
in that paper used a variant of Lemma~\ref{lem:treeaut} whose proof has an
error. We give a proof with a different constant factor.
\begin{theorem}
  Suppose $\mathcal{K}$ is a class of structures with bounded clique-width. Suppose
$\psi(\overline{x}, y)$ is a unary MSO query of input length $r$,
where all the free variables are individual variables.
Then, there exists a scalable watermarking scheme preserving
$\psi(\overline{x},y)$ on structures in $\mathcal{K}$.
\label{thm:msoTreewidth}
\end{theorem}
\begin{proof}
Suppose $G$ is a structure in $\mathcal{K}$, so it has bounded clique-width. From
Lemma~\ref{lem:cliqueDecompositionTransduction}, we get an MSO formula
$\tilde{\psi}(\overline{x},y)$, which can be interpreted on clique
decompositions of $G$ to get the same effect as interpreting
$\psi(\overline{x},y)$ on $G$. We then get an automaton $A$ equivalent
to $\tilde{\psi}(\overline{x},y)$. Let $U$ be the set of active tuples
of $G$ for the query $\psi(\overline{x},y)$. Now we apply
Lemma~\ref{lem:treeaut}, setting $\mathcal{T}$ to be a clique
decomposition of $G$ and $Y$ to be the set of active tuples $U$. We
get $n$ pairs $(b_1,b'_1),(b_2,b'_2),\ldots, (b_n,b'_n)$, where $n$
is a constant fraction of $|U|$. Given a weight function $W$ for $G$
and a mark $\mu:\{0,1\}^n$, we define the new weight function $W'$ as
follows. We set $(W'(b_i),W'(b'_i))=(W(b_i)+1,W(b'_i)-1)$ if
$\mu(i)=0$ and $(W'(b_i),W'(b'_i))=(W(b_i)-1,W(b'_i)+1)$ if
$\mu(i)=1$. For all other elements, $W'$ is same as $W$. The modified
weight function $W'$ has local distortion bounded by $1$ by
construction. The detector can recover the bits of the mark $\mu$ by
querying the original and distorted databases and noting the
differences in weights assigned to active tuples by $W$ and $W'$. We
will show that it has global distortion bounded by $r$, the input
length of $\psi(\overline{x},y)$.

Suppose $\overline{a}=\langle a_1, a_2, \ldots, a_r \rangle$ is used
as input parameter to the query $\psi(\overline{x},y)$ on $G$ and
$(b_i,b'_i)$ is a pair selected from a set $V_i$ as in
Lemma~\ref{lem:treeaut}. If $\{a_1, \ldots, a_r\} \cap V_i =
\emptyset$, then the runs of $A$ on $\mathcal{T}_{\overline{a}b_i}$
and $\mathcal{T}_{\overline{a}b'_i}$ end in the same state. Hence,
$b_i \in \psi(\overline{a},G)$ iff $b'_i \in \psi(\overline{a},G)$.
This means that either both $b_i$ and $b'_i$ are in $\psi(\overline{a},G)$
or both of them are absent. Hence, the distortion on $b_i,b'_i$ cancel
each other, provided $\{a_1, \ldots, a_r\} \cap V_i =
\emptyset$. Hence, a pair $(b_i,b'_i)$ may contribute to the global
distortion only when $\{a_1, \ldots, a_r\} \cap V_i \ne
\emptyset$. Since all the $V_i$ are mutually disjoint and there are at
most $r$ elements in $\{a_1, \ldots, a_r\}$, the global distortion is
at most $r$.
\end{proof}
Since bounded tree-width implies bounded clique-width, the above
result also holds for classes of structures with bounded tree-width.

\subsection{FO Queries on Minor Closed Structures with Locally Bounded Tree-width}
\label{sec:foLocallyBddTW}
In this section, we consider structures whose Gaifman graphs belong to
a class of graphs that has bounded local tree-width and is closed under
minors. We prove that scalable watermarking schemes exist preserving
unary first-order queries. We use concepts and techniques from
\cite{GT2004} where it is proved that in similar classes of graphs,
sets definable by unary first order formulas have bounded VC
dimension. It is observed in \cite{GT2004} that this result extends to
non-unary formulas. For this extension, \cite{GT2004} uses a
generic result from model theory that deals with VC dimension and
doesn't use Gaifman graphs. So far, there are no such
generic results about watermarking schemes yet. We have not yet found
ways to extend our results on watermarking to non-unary queries.

\subsubsection{Gaifman's Locality and Locally Bounded Tree-width}
First we review some concepts and known results that we use.
Given a structure $G = ( V, R_1^G, \ldots, R_t^G )$, its
\emph{Gaifman graph} is the undirected graph $( V, E )$,
where $(v_{1}, v_{2}) \in E$ if there is a relation $R_i$ in $G$ and a tuple
$\overline{v} \in R_i$ such that $v_1$ and $v_2$ appear in $\overline{v}$.
The distance between two elements, denoted $d(.,.)$, in a structure is
defined to be the shortest distance between them in the Gaifman graph. The
distance between two tuples of elements $\overline{v_{1}},
\overline{v_{2}}$ is $d(\overline{v_{1}}, \overline{v_{2}}) = \min\{d(v_{1},v_{2}) \mid v_{1} \in
\overline{v_{1}}, v_{2} \in \overline{v_{2}}\}$. Given $v
\in V$, $\rho \in \mathbb{N}$, the $\rho$-sphere $S_\rho(v)$ is the
set $\{ v' \mid d(v, v') \leq \rho \}$, and for a tuple $\overline{v}$,
$S_\rho(\overline{v}) = \bigcup_{v \in \overline{v}}{S_\rho(v)}$. We define
the $\rho$-neighborhood around a tuple $\overline{v}$ to be the
structure $N_\rho(\overline{v})$
induced on $G$ by $S_\rho(\overline{a})$.
The equivalence relation $\approx_\rho$ on tuples of elements
is defined as $\overline{a} \approx_{\rho} \overline{b}$ if
$N_\rho(\overline{a}) \approx N_\rho(\overline{b})$ (where $\approx$
denotes isomorphism).

A formula $\psi$ is said to be \emph{local} if there is a number $\rho
\in \mathbb{N}$ such that for every $G$ and tuples $\overline{v}_1$
and $\overline{v}_2$ of $G$, $N_\rho(\overline{v_1}) \approx
N_\rho(\overline{v_2})$ implies $G \models \psi(v_1) \iff G \models
\psi(v_2)$. This value $\rho$ is then called the \emph{locality
radius} of $\psi$. Gaifman's theorem states that every first-order
formula is local. We often annotate a formula $\psi$ with its locality
rank $r$ and write it as $\psi^{(r)}$ for the sake of explicitness.
Furthermore, $d^{>r}(v_{1},v_2)$ is a first-order formula enforcing
the distance between $v_1$ and $v_2$ to be at least $r+1$.

\begin{theorem}[Gaifman's locality theorem \cite{GAIFMAN1982}]
\label{thm:gaifman}
Every First Order formula $\varphi(\overline{x})$ is equivalent to a Boolean combination of the following:
\begin{itemize}
\item local formulas $\psi^{(\rho)}(\overline{x})$ around $\overline{x}$
    and
\item sentences of the form
    \begin{align*}
    \chi &= \exists x_1, \ldots, x_s \left
    (\bigwedge_{i=1}^s{\alpha^{(\rho)}(x_i)} \wedge \bigwedge_{1 \leq i <
    j \leq s} d^{>2\rho}(x_i, x_j) \right ) \enspace .
    \end{align*}
\end{itemize}
Furthermore,
\begin{itemize}
\item The transformation from $\varphi$ to such a Boolean combination is effective;
\item If $\mathrm{qr}(\varphi) = q$ and $n$ is the length of
    $\overline{x}$, then $\rho \leq 7^q$, $s \leq q + n$.
\end{itemize}
\end{theorem}

The \emph{$(q,k)$-type of $\overline{v}$ in $G$}, denoted by
$\mathrm{tp}^G_q(\overline{v})$, is the set of all first-order
formulas $\varphi(x_1, \ldots, x_k)$ of quantifier rank at most $q$
such that $G \models \varphi(\overline{v})$. A $(q, k)$-type
is a maximal consistent set of first-order formulas  $\varphi(x_1,
\cdots x_k)$ of quantifier rank at most $q$. Equivalently, a
$(q,k)$-type is the $(q, k)$-type of some $k$-tuple $\overline{v}$ in
some structure $G$. For a specific $(q, k)$, there are only
finitely many $(q, k)$ types. The number of such types is denoted by $t(q,k)$.

We get the following as a corollary of Gaifman's locality theorem.
\begin{corollary}
    \label{cor:typesLocality}
    Let $q \in \mathbb{N}$ and $\rho = 7^{q}$. Let $G$ be a structure
    and $\overline{a},\overline{a}' \in V^{r}$, $\overline{b},
    \overline{b}' \in V^{s}$ such that
    $\mathrm{tp}_{q}^{G}(\overline{a})=\mathrm{tp}_{q}^{G}(\overline{a}')$, 
    $\mathrm{tp}_{q}^{G}(\overline{b})=\mathrm{tp}_{q}^{G}(\overline{b}')$,
    $d(\overline{a},\overline{b}) \ge 2\rho + 1$ and
    $d(\overline{a}',\overline{b}') \ge 2\rho + 1$. Then
    $\mathrm{tp}_{q}^{G}(\overline{a},\overline{b}) =
    \mathrm{tp}_{q}^{G}(\overline{a}',\overline{b}')$.
\end{corollary}

\subparagraph*{Locally Bounded Tree-width}
We say that a class of structures $\mathcal{K}$ has
\emph{locally-bounded tree-width} if there exists a function $f : \mathbb{N}
\to \mathbb{N}$ such that given any $G \in \mathcal{K}$, any $v \in
V$ and any $r \in \mathbb{N}$, the tree-width of $N_r(x)$ is at most
$f(r)$. 

Next we recall some properties of class of graphs closed under minors.
An edge contraction is an operation which removes an edge from a graph
while simultaneously merging the two vertices it used to connect. A
graph $H$ is a minor of a graph $G$ if a graph isomorphic to $H$ can
be obtained from $G$ by contracting some edges, deleting some edges
and deleting some isolated vertices. A class $\mathcal{K}$ of graphs
is said to be closed under minors if for every graph $G$ in
$\mathcal{K}$ and every minor $H$ of $G$, $H$ is also in
$\mathcal{K}$.

Suppose a class of graphs $\mathcal{K}$ is closed under minors and has
locally bounded tree-width (the class of planar graphs is an example).
Let $G$ be a graph in $\mathcal{K}$ and let $v$ be an arbitrary vertex
in $G$. For $i \ge 0$, let $L_i$ be the set of all vertices of $G$
whose shortest distance from $v$ is $i$. For any $i,r$, the subgraph
induced by $\cup_{j=1}^{r}L_{i+j}$ on $G$ has tree-width that depends
only on $r$.
To see this, first delete all edges incident on $V'$,
where $V'=\cup_{j>i+ r}L_j$ is the set of vertices whose distance from
$v$ is more than $i+r$. Then delete all vertices in $V'$. Then
contract all edges in $V'' \times V''$, where $V''=\cup_{j\le i}L_j$
are vertices whose distance from $v$ is at most $i$. We can think of
this contracting process as merging all vertices in $V''$ into $v$.
The resulting graph $H$ is a minor of $G$, so has locally bounded
tree-width. The graph $H$ has the property that any vertex is at
distance at most $r$ from $v$, so the whole graph $H$ is in the
neighborhood of radius $r$ around $v$. Hence, the graph $H$ has
tree-width that depends only on $r$. Finally, the subgraph induced by
$\cup_{j=1}^{r}L_{i+j}$ on $G$ is a minor of $H$, so it also has
tree-width that depends only on $r$. 
This idea has been used to design
approximation algorithms for hard problems
\cite{Baker1994,Eppstein2000}.

\subsubsection{Watermarking Schemes to Preserve FO Queries on Minor
Closed Classes with Locally Bounded Tree-width}
Now we prove that there exist watermarking schemes that preserve unary
FO queries on classes of structures that are closed under minors and
that have locally bounded tree-width.
We use Gaifman's locality theorem on the FO query and consider the
constituent local queries. Answer to local queries only depend on
local neighborhoods of the structure, which have bounded tree-width.
We can run automata on them and proceed as in the previous section. We have to be careful that queries run
on overlapping neighborhoods don't interfere with each other.

Let $\mathcal{K}$ be a class of structures whose Gaifman graphs belong
to a class of graphs with locally bounded tree-width and that is
closed under minors, let $G$ be a structure in $\mathcal{K}$ and let
$\psi(\overline{x},y)$ be a unary first-order query. Let $q$ be the
rank of $\psi(\overline{x},y)$ and let $\rho$ be its locality radius.
Suppose $U \subseteq V$ is the set of active elements for the query
$\psi(\overline{x},y)$. Let $c \in U$ be an active element such that
the set $U_{c} = \{b \in U \mid \mathrm{tp}_{q}^{G}(b) =
\mathrm{tp}_{q}^{G}(c)\}$ has the maximum cardinality. Due to our
choice of $c$, we get $|U_{c}| \ge \frac{|U|}{t(q,r+1)}$ (recall that
$r$ is the length of $\overline{x}$ and $t(q,r+1)$ is the possible
number of $(q,r+1)$-types). We will show that there is a number $l$
that is a constant fraction of $|U|$ such that we can hide any mark
$\mu \in \{0,1\}^{\leq l}$. Given a weight function $W$ for $G$ and a
mark $\mu \in \{0,1\}^{l}$, we select $l$ pairs of elements
$(b_{1},b_{1}'), (b_{2}, b_{2}'), \ldots, (b_{l},b_{l}')$ from $U_{c}$
and define the new weight function $W_{\mu}$ as follows:
$(W_{\mu}(b_{i}), W_{\mu}(b_{i}'))= (W(b_{i})+1, W(b_{i})-1)$ if
$\mu(i)=1$ and $(W_{\mu}(b_{i}), W_{\mu}(b_{i}'))= (W(b_{i})-1,
W(b_{i})+1)$ if $\mu(i)=0$. For all other elements, $W_{\mu}$ is same
as $W$. The new weight function is a $1$-local distortion of the old
one by construction. The difficulty is to ensure that the global
distortion is bounded by a constant. We overcome this difficulty by
ensuring that $b_{i}$ and $b_{i}'$ cannot be distinguished by the
query $\psi(\overline{x},y)$: for almost all $\overline{a} \in V^{r}$,
$b_{i} \in \psi(\overline{a},G)$ iff $b_{i}' \in
\psi(\overline{a},G)$.
The following lemma suggests how to select such pairs.

\begin{lemma}
    \label{lem:indistinguishableLocalFormulas}
    Suppose $\psi(\overline{x},y)$ is a query and
    $\psi_{1}^{(\rho)}(\overline{x},y),
    \psi_{2}^{(\rho)}(\overline{x},y), \ldots,
    \psi_{\alpha}^{(\rho)}(\overline{x},y)$ are the local formulas given by
    Theorem~\ref{thm:gaifman} (Gaifman's locality theorem). Suppose
    $G$ is a structure and $\overline{a}\in V^{r}, b,b'\in V$ are such
    that $G\models \psi_{i}^{(\rho)}(\overline{a},b)$ iff $G\models
    \psi_{i}^{(\rho)}(\overline{a},b')$ for every $i \in \{1,2,
    \ldots, \alpha\}$. Then $b \in \psi(\overline{a},G)$ iff $b' \in
    \psi(\overline{a},G)$.
\end{lemma}
\begin{proof}
    Let $B(\psi_{1}^{(\rho)}(\overline{x},y),
    \psi_{2}^{(\rho)}(\overline{x},y), \ldots,
    \psi_{\alpha}^{(\rho)}(\overline{x},y), \chi_{1},\chi_{2}, \ldots,
    \chi_{\beta})$ be a formula equivalent to $\psi(\overline{x},y)$, as given
	by Gaifman's locality theorem (Theorem~\ref{thm:gaifman}) where $B$ is a Boolean formula
	and $\chi_{1}, \chi_{2}, \ldots, \chi_{\beta}$ are sentences that completely
	ignore the free variables $\overline{x},y$. Let $\llbracket
	\psi(\overline{a},b) \rrbracket$ denote the truth value of
	$\psi(\overline{x},y)$ in $G$ when $\overline{a}$ is assigned to $\overline{x}$
	and $b$ is assigned to $y$; similarly for other formulas.
    \begin{align*}
      \llbracket \psi(\overline{a},b)\rrbracket & =
    B(\llbracket
    \psi_{1}^{(\rho)}(\overline{a},b)\rrbracket, \llbracket
    \psi_{2}^{(\rho)}(\overline{a},b)\rrbracket, \ldots, \llbracket
    \psi_{\alpha}^{(\rho)}(\overline{a},b)\rrbracket, \llbracket
    \chi_{1}\rrbracket, \llbracket
    \chi_{2}\rrbracket, \ldots, \llbracket
    \chi_{\beta}\rrbracket)\\
      \llbracket \psi(\overline{a},b')\rrbracket & =
    B(\llbracket
    \psi_{1}^{(\rho)}(\overline{a},b')\rrbracket, \llbracket
    \psi_{2}^{(\rho)}(\overline{a},b')\rrbracket, \ldots, \llbracket
    \psi_{\alpha}^{(\rho)}(\overline{a},b')\rrbracket, \llbracket
    \chi_{1}\rrbracket, \llbracket
    \chi_{2}\rrbracket, \ldots, \llbracket
    \chi_{\beta}\rrbracket)
    \end{align*}
    From the hypothesis, $\llbracket
    \psi_{i}^{(\rho)}(\overline{a},b)\rrbracket =
    \llbracket\psi_{i}^{(\rho)}(\overline{a},b')\rrbracket$ for every
    $i \in \{1,2,\ldots, \alpha\}$. Hence the result follows.
\end{proof}

Now our goal is to select a large number of pairs $(b,b')$ from
$U_{c}$ such that they cannot be distinguished by any local query
$\psi_{i}^{(\rho)}(\overline{x},y)$, as assumed in
Lemma~\ref{lem:indistinguishableLocalFormulas}.  Let us fix some $k
\ge 1$ and apply Lemma~\ref{lem:cliqueDecompositionTransduction} to
every local query $\psi_{i}^{(\rho)}(\overline{x},y)$. We get a MSO
formula $\tilde{\psi_{i}}(\overline{x},y)$ such that for every
structure $G'$ with a parse tree $\mathcal{T}$ of clique-width at most
$k$, $\psi_{i}^{(\rho)}(G') = \tilde{\psi_{i}}(\mathcal{T})$.  Our
next goal is to identify substructures of $G$ with bounded
clique-width.  Since we are considering structures of bounded local
tree-width, any neighborhood of $G$ of bounded radius has bounded
tree-width, hence bounded clique-width.

For the MSO formulas
$\tilde{\psi_{1}}(\overline{x},y),\tilde{\psi_{2}}(\overline{x},y),
\ldots, \tilde{\psi_{\alpha}}(\overline{x},y)$, let $A_{1}, A_{2},
\ldots, A_{\alpha}$ be the corresponding tree automata. Let $A$ be the
tree automaton obtained by applying the usual product construction to
$A_{1}, A_{2}, \ldots, A_{\alpha}$ and let $m$ be the number of states
in $A$.

We pick some element $v \in V$ arbitrarily from the universe of $G$,
let $L_{0} = \{v\}$, and then define the layer $L_i$ to be the
elements of $G$ which are at a distance exactly $i$ from $v$. This
divides the graph into layers $L_0, L_1, L_2, \ldots$. For a layer
$L_j$, define the band $B_{2\rho}({L_j})$ to be the union of the
layers $L_{j-2\rho}, L_{j-{2\rho}+1}, \ldots, L_j, \ldots,
L_{j+{2\rho}-1}, L_{j+{2\rho}}$. Intuitively, $B_{2\rho}(L_{j})$
consists of the layer $L_{j}$, $2\rho$ layers to the left of $L_{j}$ and $2\rho$ layers
to the right. Let $\theta = (2 (r + 1) + 2) \rho$ and define the
band $B_{\theta}(L_i)$ in an analogous way. For $0 \leq i \leq 2\theta$,
define $\mathcal{L}_i$ to be the set of layers $\{L_i, L_{i +
2\theta+1}, \ L_{i + 4\theta+2}, \ldots \}=\{L_{i +2j\theta+j} \mid j
\ge 0\}$. Since there are $2\theta+1$
such sets, it must be the case that there is some $\mathcal{L}_i$ such
that $|U_{c} \cap (\cup \mathcal{L}_{i})| \ge
\frac{|U_{c}|}{2\theta+1}$.  We denote by $Y_1, Y_2 \ldots$ the layers
in this $\cup\mathcal{L}_i$ in increasing order of their distance from
$L_0$. If $v$ is any element in $L_{j}$, then $S_{2\rho}(v) \subseteq
B_{2\rho}(L_{j})$. Notice that by construction, $B_{2\rho}(Y_i) \cap
B_{2\rho}(Y_{j}) = \emptyset = B_{\theta}(Y_i) \cap B_{\theta}(Y_j)$
for any $i \ne j$. 
\begin{figure}
    \begin{center}
        \begin{tikzpicture}[>=stealth]
    \node[coordinate] (o1) at (0cm,0cm) {};
    \node[coordinate] (o12) at ([xshift=0.7cm]o1) {};
    \node[coordinate] (o2) at ([xshift=1.6cm]o1) {};
    \node[coordinate] (o3) at ([xshift=1.9cm]o2) {};
    \node[coordinate] (o4) at ([xshift=1.9cm]o3) {};
    \node[coordinate] (o5) at ([xshift=0.2cm]o4) {};
    \node[coordinate] (o6) at ([xshift=2.3cm]o5) {};
    \node[coordinate] (o7) at ([xshift=2.3cm]o6) {};

    \node[coordinate] (p1) at ([yshift=-0.5cm]o1) {};
    \node[coordinate] (p12) at (p1-|o12) {};
    \node[coordinate] (p2) at (p1-|o2) {};
    \node[coordinate] (p3) at (p1-|o3) {};
    \node[coordinate] (p4) at (p1-|o4) {};
    \node[coordinate] (p5) at (p1-|o5) {};
    \node[coordinate] (p6) at (p1-|o6) {};
    \node[coordinate] (p7) at (p1-|o7) {};

    \node (q1) at ([yshift=0.25cm]o1) {$L_{0}$};
    \node at ([yshift=0.25cm]o12) {$L_{1}$};
    \node at ([yshift=0.25cm]o3) {$L_{i+2j\theta+j}$};
    \node at ([yshift=0.25cm]o6) {$L_{i+2(j+1)\theta+j+1}$};


    \node[label=180:$v$,fill=black,circle,inner sep=0cm, minimum
    size=0.08cm] at (barycentric cs:o1=1,p1=1) {};

    \draw[gray!30] (o12) -- (p12);
    \draw[gray!30] (o2) -- (p2);
    \draw[gray!30] ([xshift=-0.2cm]o3) -- ([xshift=-0.2cm]p3);
    \draw (o3) -- (p3);
    \draw[gray!30] ([xshift=0.2cm]o3) -- ([xshift=0.2cm]p3);
    \draw[gray!30] (o4) -- (p4);
    \draw[gray!30] (o5) -- (p5);
    \draw[gray!30] ([xshift=-0.2cm]o6) -- ([xshift=-0.2cm]p6);
    \draw (o6) -- (p6);
    \draw[gray!30] ([xshift=0.2cm]o6) -- ([xshift=0.2cm]p6);
    \draw[gray!30] (o7) -- (p7);

    \draw[dotted,gray] ([xshift=0.2cm]barycentric cs:o12=1,p12=1) -- ([xshift=-0.2cm]barycentric cs:o2=1,p2=1);
    \draw[dotted,gray] (barycentric cs:o2=1,p2=1) --
    ([xshift=-0.2cm]barycentric cs:o3=1,p3=1);
    \draw[dotted,gray] ([xshift=0.2cm]barycentric cs:o3=1,p3=1) --
    (barycentric cs:o4=1,p4=1);
    \draw[dotted,gray] (barycentric cs:o5=1,p5=1) --
    ([xshift=-0.2cm]barycentric cs:o6=1,p6=1);
    \draw[dotted,gray] ([xshift=0.2cm]barycentric cs:o6=1,p6=1) --
    (barycentric cs:o7=1,p7=1);
    \draw[dotted,gray] ([xshift=0.2cm]barycentric cs:o7=1,p7=1) --
    ([xshift=0.5cm]barycentric cs:o7=1,p7=1);

    \draw[<->] ([yshift=-0.3cm]p2) -- node[auto=left] {\scriptsize $\theta$ layers} ([xshift=-0.2cm,yshift=-0.3cm]p3);
    \draw[<->] ([xshift=0.2cm,yshift=-0.3cm]p3) -- node[auto=left] {\scriptsize $\theta$ layers} ([yshift=-0.3cm]p4);
    \draw[<->] ([yshift=-0.3cm]p5) -- node[auto=left] {\scriptsize $\theta$ layers} ([xshift=-0.2cm,yshift=-0.3cm]p6);
    \draw[<->] ([xshift=0.2cm,yshift=-0.3cm]p6) -- node[auto=left] {\scriptsize $\theta$ layers} ([yshift=-0.3cm]p7);

    \draw[decorate,decoration=brace] ([yshift=-0.5cm]p4) --
    node[auto=left] {$B_{\theta}(L_{i+2j\theta+j})$} ([yshift=-0.5cm]p2);
    \draw[decorate,decoration=brace] ([yshift=-0.5cm]p7) --
    node[auto=left] {$B_{\theta}(L_{i+2(j+1)\theta+j+1})$} ([yshift=-0.5cm]p5);
\end{tikzpicture}
    \end{center}
    \caption{Division of Gaifman's graph of $G$ into Bands and layers.}
    \label{fig:layerConstruction}
\end{figure}
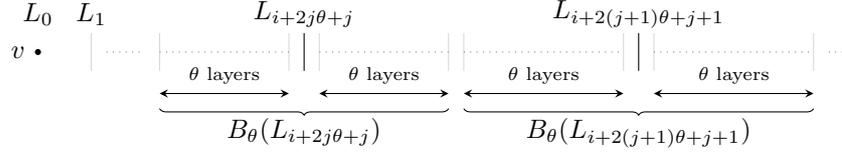
Refer to Fig.~\ref{fig:layerConstruction} for a visual representation
of the bands. The layer $L_{0}$ is represented by the single vertex
$v$. The layers $L_{i+2j\theta +j}, L_{i+2(j+1)\theta + j +1}$ are
represented by solid vertical lines. Other layers are represented by
vertical lines that are grayed out.

In the sequence of layers that we obtained, let $Y'_1, Y'_2, \ldots
Y'_{\gamma}$ be those that contain at least $2m^m+2$ elements from $U_c$.
Let $Y''_1, Y''_{2}, \ldots, Y''_{\delta}$ be the layers that contain
less than $2m^m+2$ elements from $U_{c}$. Let $v''_1, v''_2, \ldots,
v''_\delta$ be arbitrarily chosen elements of $Y''_1, Y''_{2},\ldots,
Y''_{\delta}$ respectively that
are in $U_c$ (we may ignore a particular $Y''_i$ if it does not
contain any elements of $U_c$ in it).  We will use the set of pairs
$M_{1} = \{(v''_1,v''_2), \ldots, (v''_{\delta-1}, v''_\delta)\}$ for
watermarking.

Next we select watermarking pairs from the layers $Y'_1, Y'_2, \ldots
Y'_\gamma$. For each layer $Y'_{i}$, let $N_i$ be the substructure
induced by the band $B_{\theta}(Y'_{i})$. This is a band of width
$2\theta+1$, so its tree-width and hence clique-width (say $k$) depends
only on $2 \theta+1$, which in turn depends only on the locality radius
$\rho$ and the input length $r$. Now we can apply
Lemma~\ref{lem:treeaut} with the tree automaton $A$ and the parse tree
$\mathcal{T}$ of $N_{i}$ of clique width at most $k$, setting $Y=Y'_i
\cap U_{c}$. We get pairs $(b_{(i,1)}, b'_{(i,1)}), (b_{(i,2)},
b'_{(i,2)}), \ldots, (b_{(i,n)}, b'_{(i,n)})$, where
$n=\lfloor\frac{|Y'_{i}\cap U_{c}|}{4m^m+4}\rfloor$. We will use the
set of pairs $M_{2} = \cup_{i} \cup_{j}\{(b_{(i,j)}, b'_{(i,j)})\}$
also for watermarking. Again note that all elements in the pairs are
in $U_{c}$.

\begin{lemma}
    \label{lem:noDistortionSparse}
    Suppose a watermarking pair $(v''_i, v''_{(i+1)}) \in M_{1}$
    consists of elements from $Y''_i,Y''_{(i+1)}$
    respectively. If the tuple $\overline{a}=\langle a_{1}, \ldots,
    a_{r}\rangle$ is such that $\{a_{1}, \ldots, a_{r}\} \cap (B_{2\rho}(Y''_{i})
    \cup B_{2\rho}(Y''_{(i+1)})) = \emptyset$, then $v''_{i} \in
    \psi(\overline{a},G)$ iff $v''_{(i+1)} \in \psi(\overline{a},G)$.
\end{lemma}
\begin{proof}
    We have $S_{\rho}(\overline{a}) \cap (S_{\rho}(v''_{i}) \cup
    S_{\rho}(v''_{(i+1)})) = \emptyset$. Hence we can apply
    Corollary~\ref{cor:typesLocality} to infer the result.
\end{proof}

\begin{lemma}
  \label{lem:noDistortionDense}
    Suppose a watermarking pair $(b,b') \in M_{2}$ was selected from
    some set $V_{j}$ (as specified in Lemma~\ref{lem:treeaut}) of some
    band $N_{i}$. If $\overline{a}=\langle a_{1}, \ldots,
    a_{r}\rangle$ is such that $\{a_{1}, \ldots, a_{r}\} \cap V_{j} =
    \emptyset$, $b \in \psi(\overline{a},G)$ iff $b' \in
    \psi(\overline{a},G)$.
\end{lemma}
\begin{proof}
    \emph{Case I:} $\{a_1, \ldots, a_r\} \cap B_{2\rho}(Y'_i) =
    \emptyset$. In this case, since $b,b'$ are both on the layer $Y'_i$, we
    have $S_{\rho}(\overline{a}) \cap (S_{\rho}(b) \cup  S_{\rho}(b'))
    = \emptyset$. Hence we can apply Corollary~\ref{cor:typesLocality}
    to infer the result.

    \emph{Case II:} $S_{\rho}(\overline{a}) \subseteq B_{(2(r+1)+2)\rho}(Y'_i)$. In this case,
    $S_{\rho}(\overline{a}bb') \subseteq B_{(2(r+1)+2)\rho}(Y'_i) = B_{\theta}(Y'_i)$.
    We selected $(b,b')$ according to
    Lemma~\ref{lem:treeaut}, with the tree automaton $A$ running on a
    parse tree of $N_i$. Since the tree automaton runs all the
    automata $A_{1}, A_{2}, \ldots, A_{\alpha}$ simultaneously, we
    infer that $N_i \models \psi^{(\rho)}_{j}(\overline{a},b)$ iff
    $N_i \models \psi^{(\rho)}_{j}(\overline{a},b')$ for every
    $j \in \{1,2, \ldots, \alpha\}$. Since $S_{\rho}(\overline{a}bb')
    \subseteq  B_{\theta}(Y'_i)$, the substructure induced on
    $N_i$ by $S_{\rho}(\overline{a}bb')$ is isomorphic to the
    substructure induced on $G$ by $S_{\rho}(\overline{a}bb')$. Since
    $\psi_{j}^{(\rho)}(\overline{x},y)$ is a local formula around
    $\overline{x},y$ with locality radius $\rho$, we infer that $N_i
    \models \psi^{(\rho)}_{j}(\overline{a},b)$ iff $G \models
    \psi^{(\rho)}_{j}(\overline{a},b)$ and $N_i
    \models \psi^{(\rho)}_{j}(\overline{a},b')$ iff $G \models
    \psi^{(\rho)}_{j}(\overline{a},b')$ for every
    $j \in \{1,2, \ldots, \alpha\}$. Hence, $G \models \psi^{(\rho)}_{j}(\overline{a},b)$ iff
    $G \models \psi^{(\rho)}_{j}(\overline{a},b')$ for every
    $j \in \{1,2, \ldots, \alpha\}$. We can now apply
    Lemma~\ref{lem:indistinguishableLocalFormulas} to infer the
    result.

    \emph{Case III:} $\{a_1, \ldots,a_r\} \cap B_{2\rho}(Y'_i) \neq
    \emptyset$ and $\{a_1, \ldots,a_r\} \not \subseteq
B_{(2(r+1)+2)\rho}(Y'_i)$. In this case, some elements of
$\overline{a}$ may be within distance $2\rho$ from $b, b'$. Some
elements of $\overline{a}$ may be quite far and their $\rho$
neighborhoods may not be included in $B_{\theta}(Y'_i)$. We
divide $B_{\theta}(Y'_i)\setminus B_{2{\rho}}(Y'_i)$ into
$r+1$ regions $C_1, C_2, \ldots C_{r+1}$. Define $C_1 = B_{4\rho}(Y'_i)
\setminus B_{2\rho}(Y'_i), C_2 = B_{6\rho}(Y'_i) \setminus B_{4\rho}(Y'_i)$, etc. Since
there are $r + 1$ such regions, and only $r$ parameters in
$\overline{a}$, there is a region, say $C_j$ that doesn't contain any
elements of $\overline{a}$. Let $\overline{a_{1}}$ be the tuple of
elements of $\overline{a}$ that are in $B_{2 \rho}(Y'_i) \cup C_{1}
\cup \cdots \cup C_{j-1}$ and let $\overline{a_{2}}$ consist of the
remaining elements of $\overline{a}$. Note that $S_{\rho}(\overline{a_{1}}bb') \cap
S_{\rho}(\overline{a_{2}}) = \emptyset$ (since the region $C_{j}$ is
of width $2\rho$, $\overline{a_{1}}bb'$ are on the inside of this band
and $\overline{a_{2}}$ are on the outside).
    \begin{figure}
      \begin{center}
         \begin{tikzpicture}
\node[coordinate] (o) at (0cm,0cm) {};
\node at ([yshift=0.3cm]o) {$Y_{i}'$};
\node[coordinate] (p) at ([yshift=-1.1cm]o) {};

\draw (o) -- (p);

\node[label=180:$b$,fill=black,circle,inner sep=0cm, minimum size=0.08cm] at (barycentric cs:o=3,p=1) {};
\node[label=180:$b'$,fill=black,circle,inner sep=0cm, minimum size=0.08cm] at (barycentric cs:o=1,p=3) {};

\node[coordinate] (o1) at ([xshift=0.9cm]o) {};
\node[coordinate] (p1) at (p-|o1) {};
\draw[<->] ([xshift=0.1cm]p) -- ([xshift=-0.1cm]p1);
\node at ([yshift=-0.25cm]barycentric cs:p=1,p1=1) {$2\rho$};

\draw[dashed] (o1) -- (p1);
\draw[gray!45] (barycentric cs:o=3,o1=1) -- (barycentric cs:p=3,p1=1);
\draw[gray!45] (barycentric cs:o=1,o1=3) -- (barycentric cs:p=1,p1=3);
\draw[dotted,gray] ([xshift=0.25cm]barycentric cs:o=1,p=1) -- ([xshift=-0.25cm]barycentric cs:o1=1,p1=1);

\node[coordinate] (o3) at ([xshift=0.7cm]o1) {};
\node[coordinate] (p3) at (p-|o3) {};
\draw[<->] ([xshift=0.1cm]p1) -- ([xshift=-0.1cm]p3);
\node at ([yshift=-0.25cm]barycentric cs:p1=1,p3=1) {$2\rho$};
\node at ([yshift=0.3cm]barycentric cs:o1=1,o3=1) {$C_1$};

\draw[dashed] (o3) -- (p3);
\draw[gray!45] (barycentric cs:o1=3,o3=1) -- (barycentric cs:p1=3,p3=1);
\draw[gray!45] (barycentric cs:o1=1,o3=3) -- (barycentric cs:p1=1,p3=3);
\draw[dotted,gray] (barycentric cs:o1=3,p1=3,o3=1,p3=1) -- (barycentric cs:o1=1,o3=3,p1=1,p3=3);

\node[coordinate] (o35) at ([xshift=0.8cm]o3) {};
\node[coordinate] (p35) at (p-|o35) {};
\draw[dashed] (o35) -- (p35);
\draw[dotted,gray] ([xshift=0.2cm]barycentric cs:o3=1,p3=1) -- ([xshift=-0.2cm]barycentric cs:o35=1,p35=1);

\node[coordinate] (o5) at ([xshift=0.7cm]o35) {};
\node[coordinate] (p5) at (p-|o5) {};
\draw[<->] ([xshift=0.1cm]p35) -- ([xshift=-0.1cm]p5);
\node at ([yshift=-0.25cm]barycentric cs:p35=1,p5=1) {$2\rho$};
\node at ([yshift=0.3cm]barycentric cs:o35=1,o5=1) {$C_{j-1}$};

\draw[dashed] (o5) -- (p5);
\draw[gray!45] (barycentric cs:o35=3,o5=1) -- (barycentric cs:p35=3,p5=1);
\draw[gray!45] (barycentric cs:o35=1,o5=3) -- (barycentric cs:p35=1,p5=3);
\draw[dotted,gray] (barycentric cs:o35=3,o5=1,p35=3,p5=1) -- (barycentric cs:o35=1,o5=3,p35=1,p5=3);

\node[coordinate] (o7) at ([xshift=0.7cm]o5) {};
\node[coordinate] (p7) at (p-|o7) {};
\draw[<->] ([xshift=0.1cm]p5) -- ([xshift=-0.1cm]p7);
\node at ([yshift=-0.25cm]barycentric cs:p5=1,p7=1) {$2\rho$};
\node at ([yshift=0.3cm]barycentric cs:o5=1,o7=1) {$C_{j}$};

\draw[dashed] (o7) -- (p7);
\draw[gray!45] (barycentric cs:o5=3,o7=1) -- (barycentric cs:p5=3,p7=1);
\draw[gray!45] (barycentric cs:o5=1,o7=3) -- (barycentric cs:p5=1,p7=3);
\draw[dotted,gray] (barycentric cs:o5=3,o7=1,p5=3,p7=1) -- (barycentric cs:o5=1,o7=3,p5=1,p7=3);

\node[coordinate] (o9) at ([xshift=0.7cm]o7) {};
\node[coordinate] (p9) at (p-|o9) {};
\draw[<->] ([xshift=0.1cm]p7) -- ([xshift=-0.1cm]p9);
\node at ([yshift=-0.25cm]barycentric cs:p7=1,p9=1) {$2\rho$};
\node at ([yshift=0.3cm]barycentric cs:o7=1,o9=1) {$C_{j+1}$};

\draw[dashed] (o9) -- (p9);
\draw[gray!45] (barycentric cs:o7=3,o9=1) -- (barycentric cs:p7=3,p9=1);
\draw[gray!45] (barycentric cs:o7=1,o9=3) -- (barycentric cs:p7=1,p9=3);
\draw[dotted,gray] (barycentric cs:o7=3,o9=1,p7=3,p9=1) -- (barycentric cs:o7=1,o9=3,p7=1,p9=3);

\node[coordinate] (o95) at ([xshift=0.8cm]o9) {};
\node[coordinate] (p95) at (p-|o95) {};
\draw[dashed] (o95) -- (p95);
\draw[dotted,gray] ([xshift=0.2cm]barycentric cs:o9=1,p9=1) -- ([xshift=-0.2cm]barycentric cs:o95=1,p95=1);

\node[coordinate] (o11) at ([xshift=0.7cm]o95) {};
\node[coordinate] (p11) at (p-|o11) {};
\draw[<->] ([xshift=0.1cm]p95) -- ([xshift=-0.1cm]p11);
\node at ([yshift=-0.25cm]barycentric cs:p95=1,p11=1) {$2\rho$};
\node at ([yshift=0.3cm]barycentric cs:o95=1,o11=1) {$C_{r+1}$};

\draw[dashed] (o11) -- (p11);
\draw[gray!45] (barycentric cs:o95=3,o11=1) -- (barycentric cs:p95=3,p11=1);
\draw[gray!45] (barycentric cs:o95=1,o11=3) -- (barycentric cs:p95=1,p11=3);
\draw[dotted,gray] (barycentric cs:o95=3,o11=1,p95=3,p11=1) -- (barycentric cs:o95=1,o11=3,p95=1,p11=3);

\node[coordinate] (o2) at ([xshift=-0.9cm]o) {};
\node[coordinate] (p2) at (p-|o2) {};
\draw[<->] ([xshift=-0.1cm]p) -- ([xshift=0.1cm]p2);
\node at ([yshift=-0.25cm]barycentric cs:p=1,p2=1) {$2\rho$};

\draw[dashed] (o2) -- (p2);
\draw[gray!45] (barycentric cs:o=3,o2=1) -- (barycentric cs:p=3,p2=1);
\draw[gray!45] (barycentric cs:o=1,o2=3) -- (barycentric cs:p=1,p2=3);
\draw[gray,dotted] ([xshift=-0.25cm]barycentric cs:o=1,p=1) --
([xshift=0.25cm]barycentric cs:o2=1,p2=1);

\node[coordinate] (o4) at ([xshift=-0.7cm]o2) {};
\node[coordinate] (p4) at (p-|o4) {};
\draw[<->] ([xshift=-0.1cm]p2) -- ([xshift=0.1cm]p4);
\node at ([yshift=-0.25cm]barycentric cs:p2=1,p4=1) {$2\rho$};

\draw[dashed] (o4) -- (p4);
\draw[gray!45] (barycentric cs:o2=3,o4=1) -- (barycentric cs:p2=3,p4=1);
\draw[gray!45] (barycentric cs:o2=1,o4=3) -- (barycentric cs:p2=1,p4=3);
\node at ([yshift=0.3cm]barycentric cs:o2=1,o4=1) {$C_1$};
\draw[dotted,gray] (barycentric cs:o2=3,o4=1,p2=3,p4=1) -- (barycentric cs:o2=1,o4=3,p2=1,p4=3);

\node[coordinate] (o45) at ([xshift=-0.8cm]o4) {};
\node[coordinate] (p45) at (p-|o45) {};
\draw[dashed] (o45) -- (p45);
\draw[dotted,gray] ([xshift=-0.2cm]barycentric cs:o4=1,p4=1) -- ([xshift=0.2cm]barycentric cs:o45=1,p45=1);

\node[coordinate] (o6) at ([xshift=-0.7cm]o45) {};
\node[coordinate] (p6) at (p-|o6) {};
\draw[<->] ([xshift=-0.1cm]p45) -- ([xshift=0.1cm]p6);
\node at ([yshift=-0.25cm]barycentric cs:p45=1,p6=1) {$2\rho$};

\draw[dashed] (o6) -- (p6);
\draw[gray!45] (barycentric cs:o45=3,o6=1) -- (barycentric cs:p45=3,p6=1);
\draw[gray!45] (barycentric cs:o45=1,o6=3) -- (barycentric cs:p45=1,p6=3);
\node at ([yshift=0.3cm]barycentric cs:o45=1,o6=1) {$C_{j-1}$};
\draw[dotted,gray] (barycentric cs:o45=3,o6=1,p45=3,p6=1) -- (barycentric cs:o45=1,o6=3,p45=1,p6=3);

\node[coordinate] (o8) at ([xshift=-0.7cm]o6) {};
\node[coordinate] (p8) at (p-|o8) {};
\draw[<->] ([xshift=-0.1cm]p6) -- ([xshift=0.1cm]p8);
\node at ([yshift=-0.25cm]barycentric cs:p6=1,p8=1) {$2\rho$};

\draw[dashed] (o8) -- (p8);
\draw[gray!45] (barycentric cs:o6=3,o8=1) -- (barycentric cs:p6=3,p8=1);
\draw[gray!45] (barycentric cs:o6=1,o8=3) -- (barycentric cs:p6=1,p8=3);
\node at ([yshift=0.3cm]barycentric cs:o6=1,o8=1) {$C_{j}$};
\draw[dotted,gray] (barycentric cs:o6=3,o8=1,p6=3,p8=1) -- (barycentric cs:o6=1,o8=3,p6=1,p8=3);

\node[coordinate] (o10) at ([xshift=-0.7cm]o8) {};
\node[coordinate] (p10) at (p-|o10) {};
\draw[<->] ([xshift=-0.1cm]p8) -- ([xshift=0.1cm]p10);
\node at ([yshift=-0.25cm]barycentric cs:p8=1,p10=1) {$2\rho$};

\draw[dashed] (o10) -- (p10);
\draw[gray!45] (barycentric cs:o8=3,o10=1) -- (barycentric cs:p8=3,p10=1);
\draw[gray!45] (barycentric cs:o8=1,o10=3) -- (barycentric cs:p8=1,p10=3);
\node at ([yshift=0.3cm]barycentric cs:o8=1,o10=1) {$C_{j+1}$};
\draw[dotted,gray] (barycentric cs:o8=3,o10=1,p8=3,p10=1) -- (barycentric cs:o8=1,o10=3,p8=1,p10=3);

\node[coordinate] (o105) at ([xshift=-0.8cm]o10) {};
\node[coordinate] (p105) at (p-|o105) {};
\draw[dashed] (o105) -- (p105);
\draw[dotted,gray] ([xshift=-0.2cm]barycentric cs:o10=1,p10=1) -- ([xshift=0.2cm]barycentric cs:o105=1,p105=1);

\node[coordinate] (o12) at ([xshift=-0.7cm]o105) {};
\node[coordinate] (p12) at (p-|o12) {};
\draw[<->] ([xshift=-0.1cm]p105) -- ([xshift=0.1cm]p12);
\node at ([yshift=-0.25cm]barycentric cs:p105=1,p12=1) {$2\rho$};

\draw[dashed] (o12) -- (p12);
\draw[gray!45] (barycentric cs:o105=3,o12=1) -- (barycentric cs:p105=3,p12=1);
\draw[gray!45] (barycentric cs:o105=1,o12=3) -- (barycentric cs:p105=1,p12=3);
\node at ([yshift=0.3cm]barycentric cs:o105=1,o12=1) {$C_{r+1}$};
\draw[dotted,gray] (barycentric cs:o105=3,o12=1,p105=3,p12=1) -- (barycentric cs:o105=1,o12=3,p105=1,p12=3);

\draw[<->] ([yshift=0.9cm,xshift=0.1cm]o6) -- ([yshift=0.9cm,xshift=-0.1cm]o5);
\node at ([yshift=1.15cm]barycentric cs:o6=1,o5=1) {$\overline{a_1}$ in these regions};

\draw[<->] ([yshift=0.9cm]o5) -- ([yshift=0.9cm]o7);
\node at ([yshift=1.15cm]barycentric cs:o5=1,o7=1) {$2\rho$};

\draw[<->] ([yshift=0.9cm,xshift=0.1cm]o7) -- ([yshift=0.9cm]o11);

\draw[<->] ([yshift=0.9cm]o6) -- ([yshift=0.9cm]o8);
\node at ([yshift=1.15cm]barycentric cs:o6=1,o8=1) {$2\rho$};

\draw[<->] ([yshift=0.9cm,xshift=-0.1cm]o8) -- ([yshift=0.9cm]o12);

\node at ([yshift=1.9cm]barycentric cs:o6=1,o5=1) {$\overline{a_2}$ in
these regions and outside};

\draw ([xshift=-2.5cm, yshift=1.9cm]barycentric cs:o6=1,o5=1) --
([yshift=1.9cm]barycentric cs:o8=1,o12=1) -- ([yshift=0.9cm]barycentric cs:o8=1,o12=1);
\draw ([xshift=2.5cm, yshift=1.9cm]barycentric cs:o6=1,o5=1) --
([yshift=1.9cm]barycentric cs:o7=1,o11=1) -- ([yshift=0.9cm]barycentric cs:o7=1,o11=1);

\end{tikzpicture}
      \end{center}
      \caption{The regions $C_1, \ldots, C_{r+1}$ and $\overline{a_1},\overline{a_2}$
used in Case III of the proof of Lemma~\ref{lem:noDistortionDense}.}
\label{fig:indistinguishableFO}
    \end{figure}
    Refer to Fig.~\ref{fig:indistinguishableFO} for a visual
    presentation of $\overline{a_{1}}, \overline{a_{2}}$. The layer
    $Y_{i}'$ is represented by a solid vertical line, in which
    $b,b'$ are highlighted. Other layers are
    represented by vertical lines that are grayed out. Boundaries of
    regions are represented by dashed vertical lines. Each region
    consists of $2\rho$ layers on the left and $2\rho$ layers on the
    right. Since the layer $C_{j}$ doesn't contain any elements of
    $\overline{a}$, it acts as a buffer between
    $S_{\rho}(\overline{a_{1}}bb')$ and $S_{\rho}(\overline{a_{2}})$.

    Let the structure $H_1$ be an isomorphic copy of
    $N_{i}$ (which is the substructure induced by $B_{\theta}(Y'_i)$). Since $H_1$ consists of
$2\theta+1$ layers, the tree-width and hence clique-width of
$H_1$ depends only on $2\theta+1$. Let $H_2$ be a disjoint union of at
most $r$ spheres of radius at most $2r\rho$, containing an isomorphic copy of
$N_{\rho}(\overline{a_2})$. To construct $H_{2}$,
start with an arbitrary element $a$ in $\overline{a_2}$ and include the
sphere of radius $2\rho$ around $a$. If the sphere contains any other
element $a'$ of $\overline{a_2}$, include the sphere of radius $4\rho$
around $a$, thus including the sphere of radius $\rho$ around $a'$.
Now, if the sphere of radius $4\rho$ around $a$ includes a third
element $a''$ of $\overline{a_2}$, include the sphere of radius
$6\rho$ around $a$. We need to continue this process at most $r$ times
(since there are at most $r$ elements in $\overline{a_2}$), so we get
a sphere of radius at most $2r\rho$. If this sphere doesn't contain
some element $a_3$ (this can happen if the distance between $a$ and
$a_3$ is more than $2r\rho$) of $\overline{a_2}$, we construct another sphere of
radius at most $2r\rho$ around $a_3$ to include more elements of
$\overline{a_2}$. We thus construct at most $r$ spheres of radius at most
$2r\rho$.
The clique-width of $H_2$ also depends only on $r$ and $\rho$. Let $H$ be the disjoint union of $H_1$ and
$H_2$. For the
elements
    $\overline{a_{1}}bb'$ in $N_i$, the isomorphism with $H_{1}$
    will give corresponding elements in $H_{1}$; let
    $h(\overline{a_{1}}bb')$ be these corresponding elements.
    Similarly, let $h(\overline{a_{2}})$ be the elements in $H_{2}$
    corresponding to $\overline{a_{2}}$ in
    $N_{\rho}(\overline{a_{2}})$. Let $\mathcal{T}$ be a parse tree of
    $N_i$ (and so of $H_{1}$) and
    $\mathcal{T}'$ be a parse tree of $H_{2}$ of minimum
    clique-widths, with $k$ being the maximum of these two widths.
    We obtain a parse tree $\mathcal{T}''$ of $H$ of clique-width
    at most $k$ by making $\mathcal{T}$ and $\mathcal{T}'$ as subtrees
    of a new root labeled by the union operation. We selected $b,b'$
    according to Lemma~\ref{lem:treeaut} with the tree automaton $A$
    and parse tree $\mathcal{T}$. We infer that the automaton $A$
    labels the roots of $\mathcal{T}_{h(\overline{a_{1}}b)}$ and
    $\mathcal{T}_{h(\overline{a_{1}}b')}$ with the same state.  Hence,
    the automaton $A$ labels the roots of
    $\mathcal{T}''_{h(\overline{a_{1}}b\overline{a_{2}})}$ and
    $\mathcal{T}''_{h(\overline{a_{1}}b'\overline{a_{2}})}$ with the
    same state (note that $h(\overline{a_{1}}bb')$ are in
    $\mathcal{T}$ while $h(\overline{a_{2}})$ are in
    $\mathcal{T}'$).
    Hence, we infer that $H\models
    \psi_{j}^{(\rho)}(h(\overline{a}),h(b))$
    iff $H\models\psi_{j}^{(\rho)}(h(\overline{a}),h(b'))$ for every $j \in
    \{1, 2, \ldots, \alpha\}$. The substructure induced on $G$ by
    $S_{\rho}(\overline{a}bb')$ is isomorphic to the
    substructure induced on $H$ by $S_{\rho}(h(\overline{a}bb'))$.
    Since $\psi_{j}^{(\rho)}(\overline{x},y)$ is a local formula
    around $\overline{x},y$ with locality radius $\rho$, we infer that
    $H\models \psi_{j}^{(\rho)}(h(\overline{a}),h(b))$ iff $G\models
    \psi_{j}^{(\rho)}(\overline{a},b)$ and
    $H\models\psi_{j}^{(\rho)}(h(\overline{a}),h(b'))$ iff
    $G\models\psi_{j}^{(\rho)}(\overline{a},b')$. Hence, $G \models
    \psi^{(\rho)}_{j}(\overline{a},b)$ iff $G \models
    \psi^{(\rho)}_{j}(\overline{a},b')$ for every $j \in \{1,2,
    \ldots, \alpha\}$. We can now apply
    Lemma~\ref{lem:indistinguishableLocalFormulas} to infer the
    result.
\end{proof}
The technique of considering a small number of layers to get graphs of
bounded tree-width was known before \cite{Baker1994,Eppstein2000}. Here, we find ways of
using it to bound global distortions and that is the main
technical contribution of this paper.
Now we state the main result of this sub-section.
\begin{theorem}\label{thm:presone}
  Suppose $\mathcal{K}$ is a class of structures whose Gaifman graphs
belong to a class of graphs that is closed under minors and have
locally bounded tree-width.
Suppose $\psi(\overline{x}, y)$ is a
  unary first-order query of input length $r$. Then, there exists a
  scalable watermarking scheme preserving $\psi(\overline{x},y)$ on structures in $\mathcal{K}$.
\end{theorem}
\begin{proof}
    Given a structure in $\mathcal{K}$, we select the set of
    watermarking pairs $M_{1}$ and $M_{2}$ as explained above. We can
    hide any message of $|M_{1} \cup M_{2}|$ bits using these pairs.
    For every layer $Y'_{i}$, $M_{2}$ contains $\lfloor\frac{|Y'_{i}\cap
    U_{c}|}{4m^m+4}\rfloor$ watermarking pairs. For every two layers $Y''_{j},
    Y''_{j+1}$, $M_{1}$ contains a watermarking pair. The layers
    $Y''_{j}, Y''_{j+1}$ together contain at most $4m^m+4$ elements from
    $U_{c}$. The layers $Y'_{1}, Y'_{2}, \ldots, Y'_{\gamma},
    Y''_{1}, Y''_{2}, \ldots, Y''_{\delta}$ together contain
    $\frac{U_c}{2\theta+1}$
    elements of $U_{c}$. Hence, $M_{1} \cup M_{2}$ contain at least
    $\frac{|U_{c}|}{(2 \theta +1) \cdot (4m^m+4)}$ pairs, which is the length of the message we
    can hide. Recall that $U_{c}$ itself is of cardinality at least
    $\frac{|U|}{t(q,r+1)}$. Hence, we can hide messages of at least
    $\frac{|U|}{(2 \theta+1) \cdot (4m^m+4)t(q,r+1)}$ bits.
Since $m$, $t(q,r+1)$ and $\theta$ are
    independent of the size of the structure, this scheme can hide
    messages whose length is a constant fraction of the number of
    active tuples, so it is scalable.

    Now we prove that the global distortion is at most $r$. Suppose
    $G$ is the structure we are watermarking. A pair
    $(b,b') \in M_{1} \cup M_{2}$ contributes $1$ to the global
    distortion of the query $\psi(\overline{x}, y)$ with input
    parameter $\overline{a} = \langle a_{1}, \ldots, a_{r}\rangle$ iff
    $b\in \psi(\overline{a},G)$ and $b'\notin \psi(\overline{a},G)$ or
    vice versa. For pairs $(b,b') \in M_{1}$, this can happen only if
    $\{a_{1}, \ldots, a_{r}\} \cap (B_{2\rho}(Y''_{i}) \cup
B_{2\rho}(Y''_{(i+1)})) \ne
    \emptyset$, as shown in Lemma~\ref{lem:noDistortionSparse}, where
    $Y''_{i},Y''_{(i+1)}$ are the layers containing $b,b'$ respectively.
    For pairs $(b,b') \in M_{2}$, this can happen only if $\{a_{1},
    \ldots, a_{r}\} \cap V_{j'} \ne \emptyset$, as shown in
    Lemma~\ref{lem:noDistortionDense}, where $V_{j'}$ is the
    subset of $B_{\theta}(Y''_i)$ from which the pair $(b,b')$ was chosen. The
    sets $B_{\theta}(Y''_i)$ and
    $V_{j'}$ are all mutually disjoint for
    all $i$ and $j'$. Hence, $\{a_{1}, \ldots, a_{r}\}$ can intersect
    with at most $r$ of the bands $B_{2\rho}(Y''_{i})$ and sets $V_{j'}$. So the global
    distortion of $\psi(\overline{a},G)$ is at most $r$.
\end{proof}

\section{Conclusion}
In \cite{KZ2000}, there is a transformation of
watermarking schemes for non-adversarial models into schemes for
adversarial models, under some assumptions. As observed in
\cite{Gross2011}, the same transformation under similar assumptions
also work for MSO and FO queries. Hence, our result on FO queries can
also use a similar transformation to work on adversarial models.

The difficulty with non-unary queries is that Gaifman graphs don't
capture information about active tuples --- even if two elements
$b_{1},b_{2}$ appear in the same active tuple, the Gaifman graph may
not have an edge between $b_{1}$ and $b_{2}$. The results on VC
dimension use powerful results from model theory \cite{Shelah1971} or
versions of finite Ramsey theorem for hyper graphs
\cite{Laskowski1992}. It remains to be seen whether similar
results are true for watermarking schemes.
It also remains to be seen if the condition
on closure under minors can be dropped and watermarking schemes can
still be obtained, as shown for VC dimension in \cite{GT2004}.

Beginning with graphs of bounded degree, it is now known that for the
much bigger class of graphs that are nowhere dense, FO properties can
be efficiently decided. It remains to be seen whether results on
watermarking schemes can be extended to the class of
graphs that are nowhere dense.

We don't know if there are deeper connections between bounded VC
dimension and presence of scalable watermarking schemes preserving
queries. Some progress is made in \cite{Gross2011}, where it is shown
that unbounded VC dimension doesn't necessarily mean absence of
scalable watermarking schemes, but more work is needed in this
direction.

\bibliography{references}

\newpage
\appendix
\section{{Counterpart of Lemma~\ref{lem:treeaut} in \cite{Gross2011}}}
\label{app:errorInGross2011}

In \cite{Gross2011}, the following variant of Lemma \ref{lem:treeaut} appears along with a proof. We quote the statement and the proof here and explain what the error is.

\begin{lemma}[{\cite[Lemma~5.5]{Gross2011}}]
  Let $B$ be $\Sigma_2$-tree automaton with $m$ states. Suppose $T$ is a $\Sigma$-tree and $S \subseteq T$, then there exists $n = \frac{|T|}{4m}$ distinct sets $V_1, \ldots V_n \subseteq T$ and $n$ distinct pairs $(b_i, b_i') \in V_i^2$ of distinct elements such that $\forall i \neq j, V_i \cap V_j = \emptyset$, and $\forall a \in T$:
  $$a \not \in V_i \implies \left ( b_i \in B(a, T) \iff b_i' \in B(a, T) \right ) $$
\end{lemma}

\begin{proof}
  The first phase of the proof is similar to the one that appears in Lemma \ref{lem:treeaut}.

  First, we form subtrees $U_1, U_2, \ldots U_{2n}$ by iteratively considering minimal subtrees with at least $2m$ elements of $S$ and removing them. Since $T$ is binary, each $U_i$ contains at most $4m$ elements of $S$. The number of times this construct can be iterated is $2n = \lfloor \frac{S \cap T}{4m} \rfloor$ times.

  Now, define the binary relation $F$ on $H = \{ U_1, \ldots , U_{2n} \}$ to be
the set of all pairs $(U_i, U_j)$ such that $\mathrm{lca}(U_i) \prec^T \mathrm{lca}(U_j)$ such
that there is no $k$ such that $\mathrm{lca}(U_i) \prec^T \mathrm{lca}(U_k) \prec^T \mathrm{lca}(U_j)$.
This turns $(H,F)$ into a forest with $2n$ vertices and at most $2n - 1$ edges.
Therefore, there are at most $n$ elements of this forest with more than $1$
child.

  If $U_i$ has no child, let $V_i$ be the elements of the subtree rooted at
$\mathrm{lca}(U_i)$. If $U_i$ has a child $U_j$, then define $V_i$ to be the elements of
the subtree rooted at $\mathrm{lca}(U_i)$ that are not contained in the subtree rooted
at $\mathrm{lca}(U_j)$. This ensures that $V_1, \ldots V_n$ are pairwise disjoint.

  If for some $i$, $U_i$ has no child, since $|U_i| \geq 2m$ is greater than the number $m$ of automaton states, there exists two distinct elements $b_i, b_i' \in U_i$ such that for all $a \not \in V_i$, automaton $B$ running on $T_{ab_i}$ or $T_{ab_i'}$ reaches $lca(U_i)$ in the same state.

  If $U_i$ has a child $U_j$, and $q_1, \ldots q_m$ are the states of $B$, we
define pairs $b_{i,k}, b_{i,k}'$ for $1 \leq k \leq m$ by induction on $k$.
Suppose $1 \leq k \leq m$ and that $b_{i,l}, b_{i,l}'$ are defined for $l < k$.
Since $|U_i| \geq 2m$, we have $|U_i \setminus \{ b_{i,1}, \ldots b_{i, k-1} \}
| > m$. Therefore, there exists elements $b_{i,k}, b_{i,k}' \in U_i \setminus
\{ b_{i,1}, \ldots b_{i, k-1} \}$ such that there is a state $q_{i,k}$ of $B$
for which if $a \not \in V_i$, the automaton $B$ running on either
$T_{ab_{i,k}}$ or $T_{ab_{i,k}'}$ reaches $\mathrm{lca}(U_i)$ in the same state
$q_{i,k}$ as long as the automaton leaves $\mathrm{lca}(U_j)$ in state $q_k$.

  Thus, if $U_i$ has no child and $a \in V_i$, $B$ accepts $T_{ab_i}$ iff $B$
accepts $T_{ab_i'}$. If $U_i$ has one child $U_j$, $a \not \in U_i$ and $B$
ends in $\mathrm{lca}(U_j)$ in state $q_t$, $B$ accepts $T_{ab_{i,t}}$ iff it accepts
$T_{ab_{i,t}'}$.

\end{proof}

The error is in the case where $U_i$ has one child $U_j$. In this
case, even if $a\notin V_i$, $a$ can be in $V_j$. Hence, for different
$a$, $B$ may reach $\mathrm{lca}(U_j)$ with different states $q_t$, so
accordingly we will have to choose different pairs $(b_{i,t},
b_{i,t}')$. However, what we need is a single pair $(b_i,b'_i)$ that
works for all $q_{t}$.

The proof in \cite{Gross2011} was motivated by the proof of
\cite[Lemma~7]{GT2004}. However, \cite[Lemma~7]{GT2004} is used in the
context of proving bounded VC dimension, where it is fine to have
different pairs $(b_{i},b_{i}')$ for different parameters
$\overline{a}$. Hence, the proof of \cite[Lemma~7]{GT2004} works
correctly as needed in \cite{GT2004}. In our context of designing
watermarking schemes, the same proof doesn't work and we have modified
the proof of Lemma~\ref{lem:treeaut} to work as necessary.

\end{document}